\documentclass[a4paper,unpublished,accepted=2021-03-09]{quantumarticle}
\pdfoutput=1
\usepackage[braket,qm]{qcircuit}
\usepackage[numbers,sort&compress]{natbib}

\AtBeginDocument{%
    \newwrite\bibnotes
    \def\bibnotesext{Notes.bib}
    \immediate\openout\bibnotes=\jobname\bibnotesext
    \immediate\write\bibnotes{@CONTROL{REVTEX41Control}}
    \immediate\write\bibnotes{@CONTROL{%
    apsrev41Control,author="08",editor="1",pages="1",title="0",year="0"}}
     \if@filesw
     \immediate\write\@auxout{\string\citation{apsrev41Control}}%
    \fi
  }%

\usepackage{amsthm}

\newcommand{\nt}{{\sc not}}
\newcommand{\nts}{\nt s}
\newcommand{\cnot}{{\sc C-not}}
\newcommand{\cnots}{\cnot s}
\newcommand{\ot}{\otimes}

\newcommand{\eps}{\varepsilon}

\newcommand{\cO}{\mathcal{O}}
\newcommand{\cN}{\mathcal{N}}
\newcommand{\ii}{\mathrm{i}}
\newcommand{\e}{\operatorname{e}}
\newcommand{\SP}{\operatorname{SP}}
\newcommand{\Piv}{\operatorname{Piv}}
\newcommand{\SSP}{\operatorname{SSP}}
\newcommand{\Dec}{\operatorname{Dec}}
\newcommand{\ed}{\operatorname{ed}}
\newcommand{\envel}{\operatorname{env}}
\newcommand{\elim}{\operatorname{elim}}
\newcommand{\nnz}{\operatorname{nnz}}
\renewcommand{\Re}{\operatorname{Re}}

\usepackage{mathtools}
\DeclarePairedDelimiter{\ceil}{\lceil}{\rceil}

\newtheorem{theorem}{Theorem}

\newtheorem{corollary}[theorem]{Corollary}
\newtheorem{lemma}[theorem]{Lemma}

\theoremstyle{definition}
\newtheorem{definition}[theorem]{Definition}
\newtheorem{remark}[theorem]{Remark}
\newtheorem{algorithm}{Algorithm}

\usepackage{hyperref}
\usepackage{xcolor}
\definecolor{mylinkcolor}{rgb}{0,0,0.8}
\definecolor{citecol}{rgb}{0.1,0.55,0.1}
\hypersetup{unicode=true,%
   bookmarksnumbered=false,bookmarksopen=false,bookmarksopenlevel=1, %
   breaklinks=true,pdfborder={0 0 0},colorlinks=true}%
\hypersetup{%
   anchorcolor=mylinkcolor,citecolor=citecol, %
   filecolor=mylinkcolor,linkcolor=mylinkcolor, %
   menucolor=mylinkcolor,runcolor=mylinkcolor, %
   urlcolor=mylinkcolor}

\usepackage{graphicx}
\usepackage{bm}

\begin{document}


\title{Quantum Circuits for Sparse Isometries}

\author{Emanuel Malvetti} \email{emanuel.malvetti@tum.de} \affiliation{Department of Chemistry, Technische Universit\"at M\"unchen, Lichtenbergstra{\ss}e 4, 85747 Garching, Germany}
\author{Raban~Iten} \email{itenr@itp.phys.ethz.ch} \affiliation{ETH Z\"urich, 8093 Z\"urich, Switzerland}
\author{Roger~Colbeck} \email{roger.colbeck@york.ac.uk} \affiliation{Department of Mathematics, University of York, YO10 5DD, UK}\orcid{0000-0003-3591-0576}

\date{$9^{\text{th}}$ March 2021}

\begin{abstract}
We consider the task of breaking down a quantum computation given as an isometry into \cnots{} and single-qubit gates, while keeping the number of \cnot{} gates small.  Although several decompositions are known for general isometries, here we focus on a method based on Householder reflections that adapts well in the case of sparse isometries. We show how to use this method to decompose an arbitrary isometry before illustrating that the method can lead to significant improvements in the case of sparse isometries.  We also discuss the classical complexity of this method and illustrate its effectiveness in the case of sparse state preparation by applying it to randomly chosen sparse states.
\end{abstract}

\maketitle

\section{Introduction}

A general quantum computation on an isolated system can be represented by a unitary matrix. In order to execute such a computation on a quantum computer, it is common to decompose the unitary into a quantum circuit, i.e., a sequence of quantum gates that can be physically implemented on a given architecture. There are different universal gate sets for quantum computation. Here we choose the universal gate set consisting of \cnot{} and single-qubit gates~\cite{barenco}. We measure the cost of a circuit by the number of \cnot{} gates since in many architectures they are more difficult to implement than single-qubit gates. In addition, the number of single-qubit gates is bounded by about twice the number of \cnots{}~\cite{two-qubit,two-qubit-2}, so the \cnot{} counts are illustrative of the total gate counts for this gate set.

It can also be useful to consider operations where the dimensions of the input are different to those of the output. An isometry from $m$ qubits to $n$ qubits can be represented by a $2^n\times2^m$ matrix $V$ satisfying $V^\dagger V=I$. Unitaries and state preparation are special cases of isometries where $m=n$ or $m=0$ respectively. An isometry with $m\neq n$ can be implemented by extending it to a unitary and implementing the unitary instead. The freedom in the extension can be exploited to lower the number of gates required. The main aim of the present paper is to consider decompositions that adapt well to the case of sparse isometries, i.e., those with many zero entries in the computational basis.

We briefly summarize previous work on decompositions using the gate set consisting of \cnot{} and single-qubit gates.  Arbitrary unitaries can be decomposed using $\frac{23}{48} 4^n$ \cnots{}~\cite{diagonal} to leading order, about twice as many as the best known lower bound. The most efficient known method for preparing arbitrary states requires about $2^n$ \cnots{} to leading order~\cite{plesch,ucg,isometries}, which is about twice the best known lower bound~\cite{plesch}. The decomposition of arbitrary isometries has been considered in~\cite{knill, isometries}. Near optimal methods for decomposing arbitrary isometries exist, and again they achieve \cnot{} counts approximately twice as large as the best known lower bounds. The implementation of quantum channels has been considered in~\cite{channels} and these have been implemented along with POVMs and instruments in~\cite{compiler}.

Particular classes of sparse unitaries have been studied in previous work, e.g., diagonal gates~\cite{diagonal}, uniformly controlled single-qubit gates~\cite{ucg} and permutation gates~\cite{reversible}. The case of efficiently computable sparse unitaries with a polynomial number of non-zero entries per row or column has also been considered in Ref.~\cite{sparse-unitaries}. The authors showed that these can be implemented using a polynomial number of gates, although explicit decompositions were not given.

In this work we present methods for decomposing isometries based on Householder decompositions, whose significance for (dense) circuit decompositions has been studied previously. Ref.~\cite{vadym} showed that $n$ qubit unitaries to within approximation $\eps$ can be decomposed using the Clifford+T gate library with $\cO(4^nn(\log(1/\eps)+n))$ gates. Ref.~\cite{sparsedays} gave a decomposition of unitaries with a speed advantage over methods based on the Cosine-Sine decomposition~\cite{diagonal}, although using $4^n$ \cnots{} asymptotically, rather than $\frac{23}{48}4^n$.

For dense isometries our method achieves \cnot{} counts that are close to those of~\cite{knill, isometries} (which are near optimal).  The advantage of using Householder decompositions becomes apparent when applied to sparse isometries, leading to \cnot{} counts that depend on the number of non-zero entries and the positions of the non-zero elements.  Our Householder-based algorithm proceeds one column at a time, replacing each successive column by a computational basis vector.  For sparse isometries, the improvement we get over other such decompositions is due to the relatively small amount of fill-in that occurs in each step. By contrast, using the decompositions of~\cite{knill,isometries}, which also work one column at a time, the operations performed to replace the first column typically remove all of the sparseness elsewhere.

In the case of state preparation of an $n$ qubit state $\ket{v}$ our decompositions require $\cO(n\nnz(v))$ \cnots{}, where $\nnz(v)$ is the number of non-zero entries in the state.  The counts for isometries beyond state peparation are more complicated because fill-in occurs as the algorithm proceeds, and the amount of fill-in depends on the positions of the non-zero entries.  We have three different methods for a sparse isometry $W$ from $m$ qubits to $n$ qubits, the simplest to count requires $\cO(n\nnz(W))$ \cnots{}.  Table~\ref{tab:explicit-iso-counts} summarizes our main results and gives more explicit counts. The main decompositions presented in this work have also been implemented using {\sc UniversalQCompiler}~\cite{compiler}.

\begin{table*}
\centering
\begin{tabular}{l|l|l|l}
Gate (method) & Ancillas & \cnot{} count & Reference  \\[1ex]
\hline
  State preparation&0 (1 dirty if $s=n-1$) & $(n+16s-9)\nnz(v)+\frac{23}{24}2^s$ & Corollary~\ref{cor:ssp}\\[1ex]
  State preparation&$\ceil{\frac{s}{2}-1}$ clean & $(n+6s-7)\nnz(v)+\frac{23}{24}2^s$ & Corollary~\ref{cor:ssp}\\[1ex]
  Isometry (basic)  &1 dirty & $(17n-5)\elim(W,\rho,\sigma)+(51n+34m-44)2^m$ & Remark~\ref{rmk:sparse-iso} \\[1ex]
Isometry (basic)  & $\ceil{\frac{n-3}{2}}$ clean & $(7n-3)\elim(W,\rho,\sigma)+(21n+24m-38)2^m$ & Lemma~\ref{sparse-iso} \\[1ex]
Isometry (fixed envelope)  & 1 dirty & $4\ed(\Pi_\rho W \Pi_\sigma)+\cO(n2^n)$ & Remark~\ref{rmk:fixed-env} \\[1ex]
Isometry (no fill-in)  & 1 clean + 1 dirty & $(17n + 12)\nnz(W)+(34n+34m-5)2^m$ & Lemma~\ref{lemma:no-fill-in} \\[1ex]
Isometry (no fill-in)  & $\ceil{\frac{n}{2}}$ clean& $(7n + 4)\nnz(W)+(14n+24m-21)2^m$ & Lemma~\ref{lemma:no-fill-in} \\[1ex]
\hline
\end{tabular}
\caption{{\bf Main results of this work.} This table gives \cnot{} counts (upper bounds) for the decompositions for a sparse state of $n$ qubits, $v$, or sparse isometry from $m$ to $n$ qubits, $W$. For sparse isometries we have three methods (basic, fixed envelope or no fill-in).  Here $\nnz(\cdot)$ is the number of non-zero entries of a state or of an isometry in the computational basis and $s=\ceil{\log_2 \nnz(v)}$. 
The number of eliminations $\elim(W,\rho,\sigma)$ when using the elimination order given by the row and column permutations $\rho$ and $\sigma$ is defined in Eq.~\eqref{eq:elim} and depends on the positions of the non-zero entries in the isometry, which affects how many zero elements become non-zero as the algorithm proceeds.
The quantity $\ed(W)$ is defined in Definition~\ref{def:ed}, and counts the number of elements (zero and non-zero) between the matrix envelope (as defined in Definition~\ref{def:env}) and the diagonal, and it also depends on the positions of the non-zero entries in the isometry. By $\Pi_\rho$ and $\Pi_\sigma$ we denote the permutation matrices corresponding to the permutations $\rho$ and $\sigma$. By Lemma~\ref{lemma:elim-ed-bound} we have $\elim(W,\rho,\sigma)\leq\ed(\Pi_\rho W\Pi_\sigma)$, and it trivially holds that $\ed(\Pi_\rho W\Pi_\sigma)\leq 2^{m+n}$, but our decompositions are most useful when $\elim(W,\rho,\sigma)\ll 2^{m+n}$ or $\ed(\Pi_\rho W\Pi_\sigma)\ll 2^{m+n}$.  An ancillary qubit is called clean if it starts in a known computational basis state and is restored to that state after the computation. It is called dirty if it starts in an unknown state and is restored after the computation.}
\label{tab:explicit-iso-counts}
\end{table*}

The remainder of this paper is organized as follows. In Section~\ref{sec:ssp} we discuss the case of state preparation, which is a building-block for the later cases. The main idea behind the decomposition there is to use pivoting gates, which permute the entries of the state such that the non-zero entries are grouped together, in effect reducing to state preparation on a smaller system. In Section~\ref{sec:Householder reflection} we show that Householder reflections with respect to a hyperplane orthogonal to a sparse state $\ket{v}$ can be implemented using sparse state preparation, again using at most $\cO(n\nnz(v))$ \cnots{}. In Section~\ref{sec:hdec} we show how a general isometry $V$ from $m$ to $n$ qubits can be implemented using Householder reflections before moving to the sparse case, illustrating the advantage of Householder reflections for the latter. In Section~\ref{sec:classical-complexity} we show how the main decompositions of this paper can be implemented and derive the classical complexities of the implementations.  Finally, in Section~\ref{sec:numerical-results} we present \cnot{} counts for sparse state preparation found by explicit use of our method, indicating the improvements compared to dense state preparation methods.

\section{\label{sec:ssp}Sparse State Preparation}

State preparation is the main building-block of the decompositions used in this work. In this section we introduce a method for implementing sparse state preparation more efficiently than is possible in the dense case. The main results from this section are summarized in Table~\ref{tab:explicit-counts}

\begin{table*}
\centering
\begin{tabular}{l|l|l|l|l}
Gate  & Ancillas &  Notation & \cnot{} count & Reference  \\[1ex]
\hline
 Pivoting & 0 (1 dirty if $s=n-1$)& $\Piv_v$ & $(n+16s-9)\nnz(v)$ & Lemma~\ref{lemma:pivot}\\[1ex]
Pivoting & $\ceil{\frac{s}{2}-1}$ clean & $\Piv_v$ & $(n+6s-7)\nnz(v)$ & Lemma~\ref{lemma:pivot}\\[1ex]
Permuted diagonal isometry & 0 (1 dirty if $m=n-1$)& $\Pi_n I_{n,m} \Delta_m$ & $(n+34m-34)2^m$ & Lemma~\ref{lemma:perm-diag}\footnote{\label{ft:t}Tighter bounds follow by using the more precise counts for permutations in Appendix~\ref{app:perm}.} \\[1ex]
Permuted diagonal isometry & $\ceil{\frac{m}{2}-1}$ clean& $\Pi_n I_{n,m} \Delta_m$ & $(n+24m-32)2^m$ & Lemma~\ref{lemma:perm-diag} \\[1ex]
  Householder reflection up to $\Delta\Pi$&  1 dirty & $\Delta\Pi H_v$ & $(n+16s-5)\nnz(v)+16n$ & Lemma~\ref{lemma:hr-up-to} \\[1ex]
  Householder reflection up to $\Delta\Pi$ & $\ceil{\frac{n-3}{2}}$ clean& $\Delta\Pi H_v$ & $(n+6s-3)\nnz(v)+6n$  & Lemma~\ref{lemma:hr-up-to} \\[1ex]
\hline
\end{tabular}
\caption{{\bf \cnot{} counts (upper bounds) for the additional decompositions introduced in this work.} Here $n$ is the number of output qubits, $m$ the number of input qubits, $\nnz(\cdot)$ the number of non-zero entries of a state or of an isometry in the computational basis and $s=\ceil{\log_2 \nnz(v)}$.
Furthermore $\Pi_n$ denotes an arbitrary permutation gate on $n$ qubits, $\Delta_m$ denotes an arbitrary diagonal gate on $m$ qubits and $I_{n,m}$ denotes the $2^m$ first columns of the identity gate on $n$ qubits.
All results follow from the given reference and the results in Table~\ref{tab:counts} and other entries in the present table. Slightly different results can be derived by using different decompositions for multi-controlled \nt{} gates. An ancillary qubit is called clean if it starts in a known computational basis state and is restored to that state after the computation. It is called dirty if it starts in an unknown state which must be restored after the computation.}
\label{tab:explicit-counts}
\end{table*}

\begin{definition}
Let $\ket{v}$ be a state on $n$ qubits. We say that a unitary $\SP_v$ on $n$ qubits implements state preparation for $\ket{v}$ if
$$\SP_v\ket{0}_n=\ket{v}.$$
\end{definition}

We start by presenting a useful pivoting algorithm for permuting entries in a sparse state such that all non-zero entries are grouped together. The idea is to then perform a decomposition scheme for dense state preparation on the grouped entries, which correspond to the state of a subset of the $n$ qubits.

\begin{lemma}\label{lemma:pivot}
  Let $\ket{v}$ be a state on $n$ qubits and let $\nnz(v)$ denote the number of non-zero entries of $\ket{v}$ in the computational basis. Let $s=\ceil{\log_2 \nnz(v)}$. Then there exists a permutation gate $\Piv_v$ that disentangles $n-s$ qubits in a computational basis state, i.e., $\Piv_v \ket{v} = \ket{i}_{n-s}\ot\ket{\tilde{v}}_s$ for some $s$-qubit state $\ket{\tilde{v}}_s$, and $i\in\{0,1,\ldots,2^{n-s}-1\}$.

  Let $\cN_{\Piv}(n,s)$ denote the number of \cnots{} required for pivoting any state on $n$ qubits with at most $2^s$ non-zero entries and let $\cN_{C^n_s(X)}$ denote the number of \cnots{} required to implement an $s$-controlled \nt{} when there are a total of $n$ qubits. Then
  $$\cN_{\Piv}(n,s) \leq (n-1 + \cN_{C^n_s(X)})\nnz(v).$$
Explicit counts are given in Table~\ref{tab:explicit-counts}.
\end{lemma}

\begin{proof}
  If $s=n$ there is nothing to do so assume $s < n$. The given state can be represented as a column vector in the computational basis, where the basis states can be written in terms of $n$ binary indices $\ket{b_1b_2\ldots b_n}$, which we split into two $\ket{b_1\ldots b_{n-s}}\ket{b_{n-s+1}\ldots b_n}$. The first of these corresponds to a block of the vector.  The goal is then to move all non-zero elements of $\ket{v}$ into a single block. We achieve this by using the following algorithm.
  \begin{algorithm}\label{alg:piv}
  \quad
    \begin{enumerate}
    \item\label{st:11} If all non-zero entries are in the target block, we stop the algorithm.
    \item Pick a non-zero entry outside the target block and a zero entry inside the target block. Write the basis state of the non-zero entry as $\ket{{\bf t'}}_{n-s}\ket{{\bf r'}}_s$ and that of the zero entry as $\ket{{\bf t}}_{n-s}\ket{{\bf r}}_s$.
    \item Choose a qubit on which ${\bf t'}_{n-s}$ and ${\bf t}_{n-s}$ differ.
    \item With this as a control qubit, use at most $n-1$ \cnots{} to adjust $\ket{{\bf t'}}_{n-s}\ket{{\bf r'}}_s$ to $\ket{{\bf t''}}_{n-s}\ket{{\bf r}}_s$ such that ${\bf t''}$ and ${\bf t}$ differ only on the control qubit.
    \item Use one $s$-controlled \nt{} (controlling on $\ket{{\bf r}}_s$) to exchange $\ket{{\bf t''}}_{n-s}\ket{{\bf r}}_s$ and $\ket{{\bf t}}_{n-s}\ket{{\bf r}}_s$. Note that none of the other entries of the target block are affected by this process.
    \item Return to Step~\ref{st:11}.
    \end{enumerate}
  \end{algorithm}
At the end of this algorithm, all non-zero entries of $\ket{v}$ are in the target block. Thus we used at most $n-1$ \cnots{} and one $s$-controlled \nt{} to insert one non-zero entry into the target block. Since no non-zero entry ever leaves the target block, the claimed bound follows.
\end{proof}

\begin{remark}
  In the preceding Lemma we split the vector into blocks, where we have chosen to separate the $n-s$ most significant qubits from the $s$ remaining ones. However, any splitting into $n-s$ and $s$ qubits would work. In addition, the target block and the order in which to proceed are not fixed and none of these choices affect any of the decompositions used in this work. Making these choices in the right way can reduce the \cnot{} counts. See Section~\ref{sec:ssp_impl} for details on implementing the algorithm.
\end{remark}

\begin{remark}\label{rmk:pivcount}
  Using Table~\ref{tab:counts} the following bounds on the \cnot{} count hold if $a$ dirty ancilla are available.
  \begin{align*}
    \cN_{\Piv}(n,s)&\leq(n+16s^2-28s-3)\nnz(v)\\&\qquad\qquad\qquad\qquad\qquad\text{ for }n+a\geq s+1\\
    \cN_{\Piv}(n,s)&\leq27n2^n\text{ for }n+a\geq s+1\\
    \cN_{\Piv}(n,s)&\leq(n+16s-9)\nnz(v)\text{ for }n+a\geq s+2\\
    \cN_{\Piv}(n,s)&\leq(n+8s-13)\nnz(v)\\&\qquad\qquad\qquad\text{ for }n+a\geq s+\left\lceil\frac{s}{2}\right\rceil,s\geq5
  \end{align*}
  In the case $n=s+1$, $a=0$, there are not enough qubits available for any of the known decompositions of an $s$-controlled \nt{} and hence we can decompose it as an $s$-controlled single-qubit gate (first inequality) or note that the entire pivoting can be written as a permutation (second inequality -- see Appendix~\ref{app:perm} for a more precise bound). The first of these gives a better count for small $n$.
\end{remark}

\begin{table*}
\centering
\begin{tabular}{l|l|l|l|l}
Gate  & Ancillas &  Notation & \cnot{} count & Reference  \\[1ex]
\hline
Diagonal gate & 0& $\Delta_n$ & $2^n -2$ & \cite{diagonal}  \\[1ex]
State preparation& 0& $\SP_v$ & $\frac{23}{24}2^n-2^{\frac{n}{2}+1}+5/3$ & \cite{plesch}, \cite[Remark 5]{isometries}\footnote{For convenience of presentation we have consolidated the bounds for even and odd $n$ into a single bound.}\\[1ex]
$k$-controlled single-qubit gate & 0& $C_{k}(U)$ & $16k^2 - 28k - 2 \textrm{ if } k \geq 2 \quad$  & \cite[Theorem~4]{isometries} \\[1ex]
$k$-controlled single-qubit gate & $k-1$ clean& $C_{k}(U)$ & $6k-4$ &  \cite[Corollary~4]{toffoli-ancilla} \\[1ex]
$k$-controlled \nt{} gate & 1 dirty& $C_{k}(X)$ & $16k-8$  & Corollary~\ref{cor:toffoli}, \cite{isometries} \\[1ex]
$k$-controlled \nt{} gate & $\ceil{\frac{k}{2}-1}$ dirty & $C_{k}(X)$ & $8k-12 \textrm{ if } k \geq 5$ & \cite[Proposition 5]{toffoli-ancilla} \\[1ex]
$k$-controlled \nt{} gate & $\ceil{\frac{k}{2}-1}$ clean& $C_{k}(X)$ & $6k-6 \textrm{ if } k \geq 2$ & \cite[Proposition 4]{toffoli-ancilla} \\[1ex]
Permutation gate & 0& $\Pi_n$ & $27n2^n$  & Appendix~\ref{app:perm}, \cite{reversible}\footnote{\label{ft}See Appendix~\ref{app:perm} for tighter bounds.} \\[1ex]
Permutation gate & 1 dirty& $\Pi_n$ & $(18n-26)2^n$ & Appendix~\ref{app:perm}${}^{\text{\ref{ft}}}$\\[1ex]
\hline
\end{tabular}
\caption{{\bf Useful \cnot{} counts (upper bounds) for some basic gates.} In each case $n$ denotes the number of qubits on which the gate acts and $k$ denotes the number of control qubits. An ancillary qubit is called clean if it starts in a known computational basis state and is restored to that state after the computation. It is called dirty if it starts in an unknown state which must be restored after the computation.}
\label{tab:counts}
\end{table*}

This pivoting algorithm allows us to find a scheme with low cost for the preparation of sparse states.

\begin{corollary} \label{cor:ssp}
  Let $\ket{v}$ be a state on $n$ qubits and let $\nnz(v)$ denote the number of non-zero entries of $\ket{v}$ in the computational basis. Let $s=\ceil{\log_2 \nnz(v)}$. The number of \cnots{} required for sparse state preparation of a state of $n$ qubits with $\nnz(v)$ non-zero elements is bounded by
  \begin{align*}
\cN_{\SSP}(n,s)
    &\leq \cN^{\Delta}_{\Piv}(n,s) + \cN_{\SP}(s),
  \end{align*}
where $\cN^{\Delta}_{\Piv}(n,s)$ denotes the number of \cnot{} gates used to implement pivoting up to a diagonal gate, i.e., to implement $\Delta\Piv$ for some diagonal gate $\Delta$. Explicit counts are given in Table~\ref{tab:explicit-counts}.
\end{corollary}

\begin{proof}
It is sufficient to find a circuit that maps $\ket{v}$ to $\ket{0}_{n}$, since the inverse of this circuit implements state preparation for $\ket{v}$.  Let $\Piv_v$ be constructed as in Lemma~\ref{lemma:pivot}, then $\Delta\Piv_v \ket{v}$ has the form $\ket{i}_{n-s}\ot\ket{\tilde{v}}$ for any diagonal gate $\Delta$. Without loss of generality $i=0$. Now we merely need to apply reverse state preparation for $\ket{\tilde{v}}$ on $s$ qubits. Thus sparse state preparation can be implemented as
\begin{equation} \label{eq:ssp}
    \SSP_v = (\Delta \Piv_v)^\dagger (I_{n-s}\ot\SP_{\tilde{v}}),
\end{equation}
which gives the claimed count.
\end{proof}

\begin{remark}
  For use in our sparse state preparation decomposition, it is sufficient to decompose the $s$-controlled \nt{} gates of Lemma~\ref{lemma:pivot} up to a diagonal gate. For example the Toffoli gate (with two controls) requires six \cnots{} when implemented exactly~\cite{barenco}, but it can be implemented using only three \cnots{} when implemented up to a diagonal gate~\cite[Section~VI~B]{barenco}. We are not currently aware of schemes to decompose \nt{} gates with more controls up to diagonal, but, if these were found, our counts would be improved.
\end{remark}

\begin{remark}
A more straightforward way to implement sparse state preparation is to use two-level unitaries to eliminate the non-zero entries one by one, but this leads to higher counts than the method presented here. A similar method is used in \cite{barenco} to implement arbitrary unitaries.
\end{remark}

\section{\label{sec:Householder reflection}Generalized Householder Reflections}
Given a unit vector $\ket{v}$, the standard Householder reflection~\cite{householder} with respect to $\ket{v}$ is defined as
$$H_v=I-2\op{v}{v}\,.$$
We call $\ket{v}$ the Householder vector associated with the reflection. The generalized Householder reflection of phase $\phi$ with respect to $\ket{v}$ is defined as
$$H_v^\phi=I+(\e^{\ii\phi}-1)\op{v}{v}$$
and coincides with the standard definition if $\phi=\pi$. On certain architectures generalized Householder reflections can be implemented directly~\cite{HR1} and in a fault-tolerant way~\cite{HR2}. Standard Householder reflections can be approximated well using Clifford and T gates~\cite{vadym}. In the circuit model a state preparation scheme can be used to perform a generalized Householder reflection.

\begin{lemma} \label{lemma:sp}
Let $\SP_v$ denote a unitary implementing state preparation for the state \ket{v} and $H_0^\phi$ the Householder reflection with respect to $\ket{0}$. Then $H_v^\phi = \SP_v \cdot H_0^\phi \cdot \SP_v^\dagger$.
\end{lemma}

\begin{proof}
This can be seen by considering the action on an orthonormal basis containing $\ket{v}$.
\end{proof}

\begin{lemma}
\label{lemma:reflections}
The gate $H_0^\phi$ on $n$ qubits can be implemented using an $(n-1)$-controlled single-qubit gate. The special case $H_0$ can be implemented using the same number of \cnots{} and ancilla qubits as an $(n-1)$-controlled \nt{} gate. Explicit counts for these gates are given in Table~\ref{tab:counts}.
\end{lemma}

\begin{proof}
The gate $H_0^\phi$ is a multi-controlled single-qubit gate with $n-1$ controls and hence the \cnot{} count is as in Table~\ref{tab:counts}. If $\phi = \pi$, then using some common gates
$$X =
\begin{bmatrix}
    0 & 1 \\
    1 & 0
\end{bmatrix}, \quad
Z =
\begin{bmatrix}
    1 & 0 \\
    0 & -1
\end{bmatrix}, \quad
H = \frac{1}{\sqrt{2}}
\begin{bmatrix}
    1 & 1 \\
    1 & -1
\end{bmatrix}$$
(note the Hadamard gate $H$ has a similar notation to the Householder reflection) and
$$ -Z = XZX = X(HXH)X $$
we obtain
$$H_0 = C_{n-1}(-Z) = XH \cdot C_{n-1}(X) \cdot HX\,.$$
The \cnot{} count is thus the one of $C_{n-1}(X)$ (see again Table~\ref{tab:counts}).
\end{proof}

Given two states $\ket{v}$ and $\ket{w}$ we can construct a gate that maps $\ket{v}$ to $\e^{\ii\theta}\ket{w}$ for some real $\theta$ using a standard Householder reflection defined as
\begin{equation} \label{eq:householder-vector}
H_{v,w} = H_u, \quad\text{where }
\ket{u} = \frac{\ket{v}- \e^{\ii\theta}\ket{w}}
{\|\ket{v}- \e^{\ii\theta}\ket{w}\|}
\end{equation}
with $\theta = \pi - \arg (\ip{v}{w})$ or $\theta = 0$ if $\ip{v}{w} = 0$.
We also define the generalized Householder reflection
$$\tilde{H}_{v,w} = H_u^\phi, \quad
\ket{u} = \frac{\ket{v}-\ket{w}}
{\|\ket{v}-\ket{w}\|}, \quad
\e^{\ii\phi} = \frac{\ip{v}{w}-1}{1-\overline{\ip{v}{w}}},$$
which has the property
$$\tilde{H}_{v,w}\ket{v}=\ket{w}.$$
The motivation for these definitions and related proofs are given in Appendix~\ref{A:vector}.

\section{\label{sec:hdec}Householder Decomposition}

Our goal is to find a decomposition of any isometry $V$ from $m$ to $n$ qubits. Let $I_{n,m}$ denote the first $2^m$ columns of the identity gate on $n$ qubits. If $G = G_k \cdots G_1 G_0$ is a product of elementary gates on $n$ qubits such that $GV=I_{n,m}$, then $G^\dagger$ is an extension of $V$ to a unitary. Thus $G^\dagger$ yields an implementation of $V$ using elementary gates.

Householder reflections provide a straightforward method for implementing arbitrary isometries. Let $\ket{v_0} = V\ket{0}$ be the first column of $V$ and consider $H_{v_0,0}$, the Householder reflection mapping $\ket{v_0}$ to $\ket{0}$ up to a phase. We can reduce the first column (and row by orthogonality) of $V$ by applying the Householder reflection to the isometry, i.e., the only entry in the first row and column of $H_{v_0,0}V$ is that corresponding to $\op{0}{0}$. Using the same idea the isometry can be reduced column by column to a diagonal isometry. Applying a diagonal gate on $m$ qubits then yields $I_{n,m}$. See Figure~\ref{fig:decomposition} for a schematic representation of the decomposition. 

\begin{figure*}[!t]
\centering
$$V\overset{H_{v_0,0}}{\longmapsto}
\begin{bmatrix}
    * & 0 & 0 & 0 \\
    0 & * & * & * \\
    0 & * & * & * \\
    0 & * & * & *
\end{bmatrix}
 \overset{H_{v_1,1}}{\longmapsto}
\begin{bmatrix}
    * & 0 & 0 & 0 \\
    0 & * & 0 & 0 \\
    0 & 0 & * & * \\
    0 & 0 & * & *
\end{bmatrix}
 \overset{H_{v_2,2}}{\longmapsto}
\begin{bmatrix}
    * & 0 & 0 & 0 \\
    0 & * & 0 & 0 \\
    0 & 0 & * & 0 \\
    0 & 0 & 0 & *
\end{bmatrix}
\overset{\Delta}{\longmapsto}
\begin{bmatrix}
    1 & 0 & 0 & 0 \\
    0 & 1 & 0 & 0 \\
    0 & 0 & 1 & 0 \\
    0 & 0 & 0 & 1
\end{bmatrix}$$
\caption{{\bf Basic idea of the Householder decomposition for dense isometries.} Here $*$ represents an arbitrary complex entry. Each step reduces one column without affecting the previous columns. The rows are reduced automatically due to the orthogonality of the columns. The final diagonal gate sets the phases on the diagonal equal to one.}
\label{fig:decomposition}
\end{figure*}

Before we describe the decomposition in more detail we show how the reduction of a column via Householder reflection affects the other columns of the isometry.

\begin{lemma} \label{lemma:reduction}
Let $V$ be an isometry. Let $\ket{v_j} = V\ket{j}$ be the $j^{\mathrm{th}}$ column of $V$ and $i$ be the target row index. The Householder reflection $H_{v_j,i}$ reduces the $j^{\mathrm{th}}$ column to the $i^{\mathrm{th}}$ row (i.e., such that its only non-zero element is in the $i^{\mathrm{th}}$ row).  For $s \neq i$ and $t \neq j$ we have
$$\bra{i}H_{v_j,i}V\ket{t} = \bra{s}H_{v_j,i}V\ket{j} = 0, \quad \bra{i}H_{v_j,i}V\ket{j} = \e^{\ii\theta}$$
and
$$\bra{s}H_{v_j,i}V\ket{t} = \melem{s}{V}{t} + \e^{-\ii\theta}\frac{ \melem{s}{V}{j}\melem{i}{V}{t}}{1 + |\melem{i}{V}{j}|}$$
where $\theta =  \pi + \arg (\melem{i}{V}{j})$ or $\theta = 0$ if $\melem{i}{V}{j} = 0$.
\end{lemma}

\begin{proof}
By construction we have $H_{v_j,i}V\ket{j} = \e^{\ii\theta}\ket{i}$ and by orthogonality of the columns $\bra{i}H_{v_j,i}V\ket{t} = 0$. For the final case given in the statement, we compute
\begin{align*}
\bra{s}&H_{v_j,i}V\ket{t} \\
&= \bra{s} \left(I - 2\frac{(\ket{v_j} - \e^{\ii\theta}\ket{i})(\bra{v_j} - \e^{-\ii\theta}\bra{i})}
{2 (1 + |\ip{v_j}{i}|)}\right) V \ket{t} \\
&= \melem{s}{V}{t} - \bra{s} \frac{\ket{v_j} (-\e^{-\ii\theta})\bra{i}}
{1 + |\melem{i}{V}{j}|} V \ket{t} \\
&= \melem{s}{V}{t} + \e^{-\ii\theta}\frac{ \melem{s}{V}{j}\melem{i}{V}{t}}{1 + |\melem{i}{V}{j}|}.&\hspace{-2em}\qedhere
\end{align*}
\end{proof}

\begin{corollary} \label{corr:fill-in}
Let $s \neq i$ and $t \neq j$, then $\bra{s}H_{v_j,i}V\ket{t} \neq \melem{s}{V}{t}$ if and only if $\melem{i}{V}{t}$ and $\melem{s}{V}{j}$ are non-zero.
\end{corollary}
In other words if the $j^{\text{th}}$ column is reduced to the $i^{\text{th}}$ row then the entry of $V$ at position $s,t$ does not change if $\melem{i}{V}{t}=0$ or $\melem{s}{V}{j}=0$.

\subsection{Dense isometries}
First we consider the Householder decomposition for dense isometries. This decomposition works for any isometry, but does not take advantage of any sparseness. The idea is to reduce the columns one by one. It follows from Corollary~\ref{corr:fill-in} that previously reduced columns are not affected by subsequent Householder reflections. More precisely define $V_0 = V$ and iteratively
$$\ket{v_i} = V_i\ket{i}, \quad G_i = H_{v_i, i}, \quad V_{i+1} = G_i V_i,$$
for $i=0,\dots,2^m-1$. Then $G_{2^m-1} \cdots G_0 V = I_{n,m} \Delta_m$ where $\Delta_m$ denotes a diagonal gate on $m$ qubits. This proves the following lemma.

\begin{lemma}
Let $V$ be an isometry from $m$ to $n$ qubits. Then $V$ can be implemented using $2^m$ Householder reflections on $n$ qubits and a diagonal gate on $m$ qubits.
\end{lemma}

In any sequence of Householder reflections implemented using the construction in Lemma~\ref{lemma:sp}, gates of the form $\SP_v^\dagger\cdot\SP_w$ occur. Using the dense state preparation scheme from~\cite{plesch} one can merge some gates as described in~\cite[Theorem~1]{isometries} for Knill's decomposition~\cite{knill}. This essentially halves the \cnot{} count. The structural similarity between the Householder decomposition and Knill's decomposition suggests a generalized decomposition, which we present in Appendix~\ref{app:knill-householder}. The proofs of the following two results are given in Appendix~\ref{app:dense-householder}.

\begin{lemma} \label{lemma:dense-isometries}
Let $V$ be an isometry from $m$ to $n$ qubits with $n \geq 5$. Then $V$ can be implemented via standard Householder reflections using one dirty ancilla qubit and $\mathcal{N}_{\mathrm{iso}}(m,n)$ \cnot{} gates where
\begin{eqnarray*}
\mathcal{N}_{\mathrm{iso}}(m,n)&{} \leq {}& \tfrac{23}{24}(2^{m+n}+2^n)+\tfrac{23}{12}2^{m+\frac{n}{2}}-2^{m+\frac{n}{4}+2}\\&&\quad+(16n-23)2^m-2/3 \text{\ \ if $n$ is even},\\
\mathcal{N}_{\mathrm{iso}}(m,n)&{} \leq {}& \tfrac{115}{96}(2^{m+n}+2^n)+\tfrac{23}{12}2^{m+\frac{n-1}{2}}-2^{m+\frac{n-1}{4}+2}\\&&\quad+(16n-23)2^m-2/3 \text{\ \ if $n$ is odd}.
\end{eqnarray*}
\end{lemma}

Note that, to leading order, the counts match those of~\cite[Theorem~1]{isometries} and hence, in the case of dense isometries, the advantage of using the Householder decomposition rather than Knill's is only apparent for small cases. 

\begin{corollary} \label{coro:dense-unitaries}
Let $U$ be a unitary on $n$ qubits with $n\geq5$. Then $U$ can be implemented via standard Householder reflections using one dirty ancilla qubit and $\mathcal{N}_{U}(n)$ \cnot{} gates where
\begin{eqnarray*}
\mathcal{N}_{U}(n) &{} \leq {}& \tfrac{161}{240} 4^n +  \mathcal{O}\left(2^{3n/2}\right)\text{\ \ if $n$ is even},\\
\mathcal{N}_{U}(n) &{} \leq {}& \tfrac{23}{30} 4^n +  \mathcal{O}\left(2^{3n/2}\right)\text{\ \ if $n$ is odd}.
\end{eqnarray*}
\end{corollary}

This result improves on the \cnot{} count for a similar decomposition for dense unitaries based on Householder reflections given in~\cite{sparsedays} which achieves $4^n$ to leading order. Note however that in the case of dense unitaries there are decompositions that achive better counts \cite{diagonal}.

\subsection{Sparse isometries} \label{sec:sparse-householder}

Now we consider the Householder decomposition for sparse isometries. Again the decomposition works for any isometry, but now the number of \cnot{} gates depends on the number of non-zero entries and their structure. The method yields lower \cnot{} counts than the decomposition for dense isometries if the isometry is sufficiently sparse. We do not specify a precise number of zeros an isometry should have in order to be considered sparse, but use $W$ to denote isometries in the context of methods designed to make use of sparseness.

\subsubsection{Decomposition of sparse isometries}

The basic idea of the sparse Householder decomposition is that if the columns of $W$ are sparse, then so are the Householder vectors used in the decomposition. Therefore we can use our sparse state preparation method (Corollary~\ref{cor:ssp}) to save \cnots. The main obstacle is fill-in generated by the Householder reflections, i.e., where zero entries of the isometry are removed by the reflection. Corollary~\ref{corr:fill-in} implies that such fill in is relatively small. It can be further reduced by decomposing the columns of the isometry in a well-chosen order.

More precisely let $\rho$ be a permutation of the rows of $W$ and let $\sigma$ be a permutation of its columns. We call the pair $(\rho, \sigma)$ an elimination strategy. Then the sparse Householder decomposition proceeds as in the dense case, except that at step $i$ we reduce column $\sigma(i)$ to row $\rho(i)$. If we implement the Householder reflections up to a permutation and a diagonal gate, then, after all reductions, the isometry will be a row permutation of a diagonal isometry. More precisely define $W_0 = W$ and iteratively
$$\ket{w_i} = W_i\ket{\sigma(i)}, \quad G_i = \Delta\Pi^i H_{w_i, \tilde{\rho}(i)},\quad W_{i+1} = G_i W_i$$
for $i = 0, \ldots, 2^m-1$ and where $\Pi^i$ is a permutation gate and $\tilde{\rho}(i) = (\Pi^{i-1}\cdots  \Pi^1\Pi^0 \circ \rho)(i)$ with $\circ$ denoting the action of a permutation gate on a permutation. Then $G_{2^m-1} \cdots G_0 W = \Pi_n I_{n,m} \Delta_m$. The following two lemmas show how to decompose permuted diagonal gates and how to implement Householder reflections up to a permutation and a diagonal gate.

\begin{lemma} \label{lemma:perm-diag}
Let $\Pi_n$ be a permutation gate on $n$ qubits and $\Delta_m$ a diagonal gate on $m$ qubits. An isometry of the form $\Pi_n I_{n,m} \Delta_m$ can be implemented using
\begin{align*}
\cN_{\Pi I \Delta}(n,m)
\leq \cN_{\Piv}^\Delta(n,m) + \cN_{\Pi}^\Delta(m) + \cN_\Delta(m)
\end{align*}
\cnots{}. Explicit counts are given in Table~\ref{tab:explicit-counts}.
\end{lemma}

\begin{proof}
  Consider the vector $\ket{u}$ formed by summing the columns of $\Pi_n I_{n,m} \Delta_m$ and the pivoting algorithm (Algorithm~\ref{alg:piv}) applied to this. If the gates required to pivot $\ket{u}$ are applied to $\Pi_n I_{n,m} \Delta_m$, all the $2^m$ non-zero entries of $\Pi_n I_{n,m} \Delta_m$ are moved to the top.  [Note that when taking the counts from Remark~\ref{rmk:pivcount}, we replace $\nnz(v)$ by $2^m$.]  The resulting isometry has the form $I_{n,m} \Pi_m \Delta_m$. Since it leads to the same structure, it suffices to perform the pivoting up to a diagonal. It remains to implement a permutation on $m$ qubits (up to diagonal) and a diagonal gate on $m$ qubits.
\end{proof}

\begin{lemma}
\label{lemma:hr-up-to}
Let $\ket{v}$ be a state on $n$ qubits and let $\nnz(v)$ denote the number of non-zero entries of $\ket{v}$ in the computational basis. Let $s=\ceil{\log_2 \nnz(v)}$. Then the standard Householder reflection $H_v$ can be implemented up to diagonal and permutation gates using
\begin{align*}
\cN_H^{\Delta\Pi}(n,s)
&\leq 2 \cN_{\SP}(s) + \cN_{\Piv}^\Delta(n,s) + \cN_{H_0}(n)
\end{align*}
\cnots{}. Explicit counts are given in Table~\ref{tab:explicit-counts}.
\end{lemma}

\begin{proof}
It follows from Lemma~\ref{lemma:sp} and Eq.\ \eqref{eq:ssp} that $H_v$ can be implemented up to diagonal and permutation gates as
$$(\Delta\Piv_v)H_v=(I_{n-s}\ot\SP_{\tilde{v}})H_0(I_{n-s}\ot\SP_{\tilde{v}})^\dagger\Delta\Piv_v,$$
which gives the claimed bound.
\end{proof}

Now we can give the \cnot{} counts for the sparse Householder decomposition. First we define
\begin{equation} \label{eq:elim}
\elim(W, \rho,\sigma)=\sum_{i=0}^{2^m-1} (\nnz(w_i)-1)
\end{equation}
to be the total number of eliminations during the Householder decomposition of $W$ when using the elimination strategy $(\rho,\sigma)$. Recall that $\ket{w_i}$ denotes the column reduced in the $i^\text{th}$ step of the decomposition, which differs in general from $W\ket{\sigma(i)}$.

\begin{lemma} \label{sparse-iso}
Let $W$ be an isometry from $m$ to $n$ qubits and let $(\rho,\sigma)$ be an elimination strategy. Then $W$ can be implemented using
$$\cN_{\mathrm{iso}}(n,m)\leq\sum_{i=0}^{2^m-1} \cN_H^{\Delta\Pi}(n,s(i))+\cN_{\Pi I\Delta}(n,m)$$
\cnots{} where $s(i) = \ceil{\log_2(1+\nnz(w_i))}$. Explicit counts are given in Table~\ref{tab:explicit-iso-counts}.
\end{lemma}
\begin{proof}
In the Householder decomposition the steps involve $H_{w_i,\tilde{\rho}(i)}$. This is a standard Householder reflection with respect to a vector that may have one additional non-zero element (see Eq.~\eqref{eq:householder-vector}). The bound given then follows from the sparse Householder decomposition.
\end{proof}

The count resulting from Lemma~\ref{sparse-iso} depends on the chosen elimination strategy. The optimal strategy is the one that minimizes the amount of fill-in produced and therefore the number of eliminations required.

\begin{remark}\label{rmk:betterCU}
It can be beneficial to use the idea of the sparse Householder decomposition, without adhering to the exact form given above. For example, using a single standard Householder reflection we can implement a $k$-controlled single-qubit gate up to diagonal. This gate can be used to improve the \cnot{} count of the column-by-column decomposition~\cite{isometries}. Indeed, without loss of generality we can assume that the least significant qubit is the target qubit of the $k$-controlled single-qubit gate. Then the corresponding unitary is the identity matrix except for the $2\times 2$ block in the bottom right corner. We can reduce the penultimate column (up to diagonal) using a standard Householder reflection with respect to $\ket{1\ldots10}_n$ and two single-qubit gates for the state preparation and reverse state preparation. The \cnot{} count is then that of a $k$-controlled \nt{} gate by a similar argument as in Lemma~\ref{lemma:reflections}.
\end{remark}

\begin{remark}\label{rmk:sparse-iso}
To obtain a more explicit bound we can plug in the counts from Table~\ref{tab:explicit-counts}:
\begin{align*}
\cN_{\mathrm{iso}}&(n,m)
\leq\sum_{i=0}^{2^m-1}\left[ (n+16s(i)-5)(1+\nnz(w_i))+16n\right]\\
&\qquad\ \ \,+(n+34m-34)2^m\\
&\leq(17n-5)\sum_{i=0}^{2^m-1}\nnz(w_i) +(34n+34m-39)2^m\\
&=(17n-5)\elim(W,\rho,\sigma)+(51n+34m-44)2^m,
\end{align*}
where we have used $s(i)\leq n$. The given bound is valid with one dirty ancilla.
\end{remark}

\begin{remark}
For very sparse isometries on a small number of qubits, the implementation cost of the $H_0$ gates used to implement standard Householder reflections (see Lemma~\ref{lemma:hr-up-to}) dominates the total number of \cnots{} required to implement the sparse isometry. However, on some experimental architectures it may be possible to directly implement $H_0$ gates by evolving the quantum system with an adapted Hamiltonian (see, e.g.,~\cite{fenner}). Such architectures would hence allow for a very low cost implementation of sparse isometries on a small number of qubits.
\end{remark}

\subsubsection{Fill-in and envelopes} \label{sec:envelopes}

In order to gain as much advantage as possible from the sparseness of an isometry, we need to minimize fill-in as much as possible, which corresponds to choosing $\rho$ and $\sigma$ so as to minimize $\elim(W,\rho,\sigma)$. Reducing a column of $W$ in general affects other columns and creates new non-zero entries. Due to the orthogonality of the columns however, Householder reflections create little fill-in when applied to isometries. In fact, it follows immediately from Corollary~\ref{corr:fill-in} that when reducing column $j$ to row $i$ fill-in can only occur in columns that are non-zero in the $i^{\text{th}}$ entry and fill-in is confined to the rows where $W\ket{j}$ is non-zero.

We use matrix envelopes to give a bound on the amount of fill-in the sparse Householder decomposition produces. The envelope of a sparse matrix gives for each column of the matrix the row index of the lowest non-zero element (i.e., the largest row index of a non-zero element) in that column or any previous column.

\begin{definition} \label{def:env}
Let $W$ be an isometry from $m$ to $n$ qubits. The \emph{envelope} of $W$ is defined to be the function $\envel_W : \{0, 1, \dots, 2^m-1\} \rightarrow \{0, 1, \dots, 2^n-1\}$ that maps each column index $j$ to the smallest row index $\envel_W(j)$ such that $\melem{i}{W}{j} = 0$ for all $i > \envel_W(j)$ and such that $\envel_W(j+1) \ge \envel_W(j)$.
\end{definition}

\begin{definition} \label{def:ed}
Let $W$ be an isometry from $m$ to $n$ qubits. Define $\ed(W) = \sum_{j=0}^{2^m-1}(\envel_{W}(j)-j)$. Then $\ed(W)$ denotes the number of entries between the envelope and the diagonal of $W$.
\end{definition}

The definition of $\ed(W)$ is motivated by the following result.

\begin{lemma} \label{lemma:elim-ed-bound}
Let $W$ be a sparse isometry from $m$ to $n$ qubits and let $(\rho,\sigma)$ be arbitrary row and column permutations. Let $\Pi_\rho$ and $\Pi_\sigma$ denote the corresponding permutation matrices. Then $\elim(W,\rho,\sigma) \leq \ed(\Pi_\rho W \Pi_\sigma)$.
\end{lemma}

\begin{proof}
Let $\tilde{W}=\Pi_\rho W\Pi_\sigma$ for some permutations $(\rho,\sigma)$. Reducing the columns of $W$ according to $(\rho, \sigma)$ is equivalent in terms of fill-in to reducing $\tilde{W}$ according to the trivial elimination strategy $(\iota,\iota)$ where $\iota$ denotes the identity permutation, i.e., $\elim(W,\rho,\sigma) = \elim(\tilde{W},\iota,\iota)$. It follows from Corollary~\ref{corr:fill-in} that the fill-in for $\tilde{W}$ is confined to entries $(i,j)$ with $i\leq\envel(j)$.  Due to orthogonality of the columns, if we proceed in column order we never need to eliminate any elements above the diagonal, so $\elim(\tilde{W},\iota,\iota) \leq \ed(\tilde{W})$.
\end{proof}

Finding the row and column permutations that minimize the envelope is a computationally difficult task. Methods for similar problems have been considered in the context of sparse matrix decompositions~\cite{colqr}. Note that once a column permutation is fixed, the optimal row permutation can be found in the following straightforward way.

Let $W$ be an isometry from $m$ to $n$ qubits. Consider the following algorithm for constructing a modified isometry $W'$.

\begin{algorithm}\label{alg:1}
\quad
\begin{enumerate}
\item Set $i=0$, $j=0$ and $W'=W$.
\item \label{st:2} Set $k$ to be the number of non-zero elements in the column with index $j$ with row indices greater than or equal to $i$.
\item If $k\neq 0$, permute the rows with indices greater than or equal to $i$ such that these $k$ non-zero elements have row indices $i$ to $i+k-1$, assign the new isometry to $W'$ and set $i=i+k$.
\item If $j<2^m-1$, set $j=j+1$ and return to Step~\ref{st:2}, otherwise
  output $W'$.
\end{enumerate}
\end{algorithm}

\begin{lemma}
Let $W$ be an isometry from $m$ to $n$ qubits and $W'$ be the output after applying Algorithm~\ref{alg:1} to $W$. For all row permutations $\rho'$ we have $\envel_{W'}(j)\leq \envel_{\Pi_{\rho'} W}(j)$ for all $j$.
\end{lemma}
\begin{proof}
Given a column index $j$, let $t(j)$ be the number of rows such that
for each of the columns with indices smaller than or equal to $j$,
those rows have at least one non-zero element, i.e.,
$$t(j):=2^n-|\{i:\bra{i}W\ket{j'}=0\ \forall\ j'\in\{0,1,\ldots,j\}\}|\,.$$
It follows that for any row permutation $\rho'$ we have
$\envel_{\Pi_{\rho'} W}(j)\geq t(j)-1$ for all $j$ (recall that $\envel_{\Pi_{\rho'} W}$ is a
non-decreasing function by definition).  However, by construction,
$\envel_{W'}(j)=t(j)-1$ for all $j$, from which the claim follows.
\end{proof}

The difficult part is therefore finding the best column permutation. In this work we consider a simple greedy algorithm for finding a good column permutation.

\begin{algorithm} \label{alg:greedy-envelope}
\quad
  \begin{enumerate}
\item Set $i=0, j=0$ and $W'=W$.
\item\label{st:22} Set $M$ to be the submatrix of $W'$ formed by only considering rows with row index greater than or equal to $i$ and columns with column index greater than or equal to $j$.
\item Pick one of the columns of $M$ with the fewest non-zero elements and set $k$ to be the number of non-zero elements.
\item Permute the columns of $W'$ with column index greater than or equal to $j$ such that when restricting the permutation to $M$, the chosen column becomes the first column of $M$. Then permute the rows of $W'$ with row index greater than or equal to $i$ such that when restricting the permutation to $M$, all non-zero elements of the chosen column are moved to the top. Set $i=i+k$, $j=j+1$.
\item If $j<2^m$ and $i<2^n$ return to Step~\ref{st:22}, otherwise output $W'$.
  \end{enumerate}
  \end{algorithm}
    This algorithm corresponds to minimizing the increment of the envelope at each step. More information on an efficient way to implement it can be found in Section~\ref{sec:implementation}.

\subsection{Fixed envelope method}
The asymptotic \cnot{} count for the sparse Householder decomposition contains an undesirable factor $n$ stemming from Lemma~\ref{lemma:pivot}. We now show how this factor can be avoided if we give up some control on the amount of fill-in.

For this we make use of the decrement gate $\Dec_n$ which subtracts $1$ in the computational basis (modulo $2^n$), i.e., $\Dec_n=\sum_{i=0}^{2^n-1}\op{i}{i\oplus 1}$ where $\oplus$ denotes addition modulo $2^n$. This gate can be implemented using $\cO(n)$ \cnots{} and one ancilla qubit~\cite{decrement}. The method is illustrated in Figure~\ref{fig:modified}.

\begin{figure*}[!t]
\centering
$$W=
\begin{bmatrix}
    0 & 0 & 0 & * \\
    * & * & * & * \\
    0 & * & * & 0 \\
    * & 0 & * & 0
\end{bmatrix}
 \overset{\Pi_\rho}{\longmapsto}
\begin{bmatrix}
    * & * & * & * \\
    * & 0 & * & 0 \\
    0 & * & * & 0 \\
    0 & 0 & 0 & *
\end{bmatrix}
 \overset{H_{w_{\sigma(0)},0}}{\longmapsto}
\begin{bmatrix}
    \odot & - & - & - \\
    \times & + & * & + \\
    0 & * & * & 0 \\
    0 & 0 & 0 & *
\end{bmatrix}
\overset{\Dec_2}{\longmapsto}
\begin{bmatrix}
    0 & * & * & * \\
    0 & * & * & 0 \\
    0 & 0 & 0 & * \\
    * & 0 & 0 & 0
\end{bmatrix}$$
\caption{{\bf Illustration of the first step of the fixed envelope method.} Here $*$ represents an arbitrary complex entry, $\odot$ denotes the target entry of the reduction, $\times$ stands for an entry that was eliminated, $-$ denotes entries eliminated due to the orthogonality constraint and $+$ means that fill-in occurs. Here $\sigma(0) = 0$.}
\label{fig:modified}
\end{figure*}

\begin{lemma} \label{lemma:modified-method}
Let $W$ be a sparse isometry from $m$ to $n$ qubits and let $(\rho,\sigma)$ be some row and column permutations. Then $W$ can be implemented using
$$\cN_{\mathrm{iso}}(n,m)\leq \mathcal{N}^{\Delta}_{\Pi}(n)+\sum_{i=0}^{2^m-1}\mathcal{N}_{\mathrm{col}}(s(i))+\mathcal{N}_{\Pi}^{\Delta}(m)+\mathcal{N}_{\Delta}(m)$$
\cnots{} where $s(i) = \ceil{\log_2(1+\envel_{\Pi_\rho W \Pi_\sigma}(i)-i)}$ and where
$$\mathcal{N}_{\mathrm{col}}(s(i))=2\mathcal{N}_{\SP}(s(i))+\mathcal{N}_{H_0}(n)+\mathcal{N}_{\Dec_n}.$$
Explicit counts are given in Table~\ref{tab:explicit-iso-counts}.
\end{lemma}

\begin{proof}
We consider reducing $W$ in the following way. First apply the row permutation $\rho$ up to diagonal. Then reduce the columns in the order given by the column permutation $\sigma$. For each column, use a Householder reflection to reduce it to the topmost row. Apply the decrement gate and then move to the next column. After all columns have been reduced in this way we apply $X^{\ot n-m}\ot I_m$. The resulting isometry has the form $I_n \Pi_m \Delta_m$ which can be reduced to the identity by applying a permutation on $m$ qubits up to diagonal and a diagonal gate on $m$ qubits.

By construction, before each Householder reflection, all non-zero entries in the column being reduced are in the topmost $2^{s(i)}$ positions.  We can hence perform the Householder reflections using the method of Lemma~\ref{lemma:hr-up-to} but omitting the pivoting steps.
\end{proof}

\begin{remark} \label{rmk:fixed-env}
To obtain a more explicit bound we can plug in the counts from Table~\ref{tab:explicit-iso-counts} to obtain
$$
\cN_{\mathrm{iso}}(n,m)
\leq 4\ed(\Pi_\rho W \Pi_\sigma)+\cO(n2^n).
$$
This bound is valid with one dirty ancilla. The factor $4$ stems from the fact that each Householder reflection uses state preparation twice and each state preparation acts on $s(i)$ qubits where $2^{s(i)}$ is at most twice the height of the envelope $1+\envel_{\Pi_\rho W \Pi_\sigma}(i)-i$ in column $i$ of $\Pi_\rho W \Pi_\sigma$.
\end{remark}

\subsection{No fill-in method}

Using a clean ancilla qubit we can avoid fill-in altogether. The method is illustrated in Figure~\ref{fig:no-fill}.

\begin{figure*}[!t]
\centering
$$\tilde{W}=
\begin{bmatrix}
    0 & 0 & 0 & * \\
    * & * & * & * \\
    0 & * & * & 0 \\
    * & 0 & * & 0 \\
    0 & 0 & 0 & 0 \\
    0 & 0 & 0 & 0 \\
    0 & 0 & 0 & 0 \\
    0 & 0 & 0 & 0
\end{bmatrix}
 \overset{H_{v_0,2^n}}{\longmapsto}
\begin{bmatrix}
    0 & 0 & 0 & * \\
    \times & * & * & * \\
    0 & * & * & 0 \\
    \times & 0 & * & 0 \\
    \odot & 0 & 0 & 0 \\
    0 & 0 & 0 & 0 \\
    0 & 0 & 0 & 0 \\
    0 & 0 & 0 & 0
\end{bmatrix}
 \overset{H_{v_1,2^n+1}}{\longmapsto}
\begin{bmatrix}
    0 & 0 & 0 & * \\
    0 & \times & * & * \\
    0 & \times & * & 0 \\
    0 & 0 & * & 0 \\
    * & 0 & 0 & 0 \\
    0 & \odot & 0 & 0 \\
    0 & 0 & 0 & 0 \\
    0 & 0 & 0 & 0
\end{bmatrix}
\overset{H_{v_2,2^n+2}}{\longmapsto}
\begin{bmatrix}
    0 & 0 & 0 & * \\
    0 & 0 & \times & * \\
    0 & 0 & \times & 0 \\
    0 & 0 & \times & 0 \\
    * & 0 & 0 & 0 \\
    0 & * & 0 & 0 \\
    0 & 0 & \odot & 0 \\
    0 & 0 & 0 & 0
\end{bmatrix}
\overset{H_{v_3,2^n+3}}{\longmapsto}
\begin{bmatrix}
    0 & 0 & 0 & \times \\
    0 & 0 & 0 & \times \\
    0 & 0 & 0 & 0 \\
    0 & 0 & 0 & 0 \\
    * & 0 & 0 & 0 \\
    0 & * & 0 & 0 \\
    0 & 0 & * & 0 \\
    0 & 0 & 0 & \odot
\end{bmatrix}$$
\caption{{\bf Illustration of the no fill-in method.} The clean ancilla leads to the empty $4\times4$ block at the start. Here $*$ represents an arbitrary complex entry, $\times$ stands for an entry that was eliminated and $\odot$ denotes the target entry of each reduction. We actually implement each of the Householder reflections up to permutation, so the final state is a row permutation of that shown (for simplicity we did not depict this).}
\label{fig:no-fill}
\end{figure*}

\begin{lemma} \label{lemma:no-fill-in}
Let $W$ be a sparse isometry from $m$ to $n$ qubits. Then, using one additional clean ancilla, $W$ can be implemented using
$$\cN_{\mathrm{iso}}(n,m)\leq \sum_{i=0}^{2^m-1}\cN_H^{\Delta\Pi}(n+1,s(i)) + \cN_{\Pi I \Delta}(n+1,m)$$
\cnots{} where $s(i) = \ceil{\log_2 (1+\nnz(W\ket{i}))}$.
Explicit counts are given in Table~\ref{tab:explicit-iso-counts}.
\end{lemma}

\begin{proof}
We implement the isometry $\tilde{W}$ from $m$ to $n+1$ qubits defined by
$$\tilde{W} \ket{i} = \ket{0} \otimes W\ket{i}$$
which in the computational basis is just $W$ stacked on top of a zero matrix of the same size. Then each of the $2^m$ columns can be reduced to one of the $2^n$ zero rows without creating any fill-in. These reductions can be implemented using Householder reflections up to a diagonal and permutation gate as in Lemma~\ref{lemma:hr-up-to}. Then using Lemma~\ref{lemma:perm-diag} we reduce the resulting permuted diagonal isometry to the identity.
\end{proof}

\section{Classical Complexity}
\label{sec:classical-complexity}

In this section we compute the classical worst-case time complexity of some decompositions presented in this work. We compare the dense Householder decomposition to other known methods and we propose a sparse storage format which is well adapted to the sparse Householder decomposition.

\subsection{Dense isometries}

First we consider the dense case. For each column we have to compute the corresponding Householder vector (see Eq.\ \eqref{eq:householder-vector}) and apply the Householder reflection to the entire isometry and produce the circuit implementing the Householder reflection. Computing the vector takes $\cO(2^n)$ and applying the reflection takes $\cO(2^{m+n})$. Producing the circuit takes $\cO(2^{3n/2})$ according to Appendix B.4 of~\cite{compiler}. Thus the classical complexity is $\cO(2^{m+n}(2^m+2^{n/2}))$. For comparison, the column-by-column decomposition requires $\cO(n2^{2m+n})$, and Knill's decomposition and the Cosine-Sine decomposition both require $\cO(2^{3n})$~\cite{compiler}. A high performance implementation for a decomposition of dense unitaries based on Householder reflections is presented in~\cite{sparsedays}. 

\subsection{Sparse state preparation} \label{sec:ssp_impl}

The non-zero pattern of a sparse state on $n$ qubits with at most $2^s$ non-zero entries can be compactly stored as a list of the $n$-bit indices of the non-zero entries. This requires $\cO(n2^s)$ space. 

The pivoting algorithm, Algorithm~\ref{alg:piv}, can be implemented with some greedy optimizations. First there are $\binom{n}{s}$ possible splittings for the $n$ qubits. We will think of the vector as being reshaped into a two dimensional array with $2^s$ rows and $2^{n-s}$ columns corresponding to the chosen qubit splitting. If the number of possible splittings is small enough, we try all splittings and choose the one with the largest number of non-zero elements in one column and this will be the target column. Otherwise one can randomly sample a fixed number of splittings and use the best one. Second, in each insertion step we choose an element not yet in the target column for which the cost of inserting it into the target column is minimal and perform the insertion. We iterate this until all elements are in the target column.

The insertion of one element into the target column can be implemented using one $s$-controlled \nt{} gate and $d-1$ \cnots{}, where $d$ is the Hamming distance between the index of the non-zero entry and the index of the target entry. Thus we want to find a non-zero entry outside the target column and a zero entry in the target column for which the Hamming distance is minimal. We do this by finding for each non-zero entry outside the target column the closest zero entry in the target column. The Hamming distance can be written as $d=d_c+d_r$ where $d_c$ is the Hamming distance obtained when restricting the indices to the column indices, and $d_r$ is the part corresponding to the row indices. For each non-zero entry, we can compute $d_c$ in $\cO(n-s)$, which adds up to $\cO((n-s)2^s)$ for all non-zero entries outside the target column. We can compute $d_r$ for all non-zero entries at the same time in  $\cO(s2^s)$ by using breadth first search on the $s$-dimensional hypercube with multiple starting vertices, given by the row indices of the zero entries in the target column. More precisely, we store a list of length $2^s$, where each entry corresponds to one row and stores the minimal distance $d_r$ to some free row of the target column. The entries corresponding to free rows are initialized with distance $0$ and all other entries are initialized with distance $\infty$. Then we perform the usual breadth first search on the graph whose vertices are given by the entries of the list and whose edges connect any pair of entries whose indices have Hamming distance one. Performing the insertion takes $\cO(n2^s)$. The entire circuit implementing sparse state preparation with the greedy optimizations mentioned above can thus be computed in time $\cO(\binom{n}{s} + n2^{2s})$.

\subsection{Sparse isometries} \label{sec:implementation}

To store a sparse isometry $W$ we store two arrays $R$ and $C$ of size $2^n$ and $2^m$ respectively. For a given row index $i$, $R(i)$ stores a reference to a balanced tree (e.g. a red-black tree) containing for each non-zero element of the $i^{\text{th}}$ row a triplet of the form $(i,j,\melem{i}{W}{j})$. The elements of the tree are sorted according to the key $j$. Analogously we define $C(j)$ to store a reference to a tree containing the non-zero elements of the $j^{\text{th}}$ column. This requires $\cO(2^n+2^m+\nnz(W))$ space. Given row and column indices $i$ and $j$, the corresponding entry can be created, read, modified or deleted in time $\cO(n)$.

We now show how to reduce column $j$ to row $i$. From Corollary~\ref{corr:fill-in} we know that the modified entries are those in column $j$ and row $i$ and those with indices $s$ and $t$ such that $s \in C(j)$ and $t \in R(i)$. First we iterate over all choices for $s \neq i$ and $t \neq j$ and create or modify the entries according to Lemma~\ref{lemma:reduction}. Then we set the entries in column $j$ and row $i$ to zero except for the entry with indices $(i,j)$ which is determined by Lemma~\ref{lemma:reduction}. Thus each reduction can be done in time $\cO(n\,\mathrm{mod})$ where $\mathrm{mod}$ denotes the number of modified elements.

In Algorithm~\ref{alg:greedy-envelope} we presented a greedy method for constructing a permutation of an isometry leading to a small envelope. For a sparse isometry $W$ given in the data structure described above, the corresponding row and column permutations $\rho$ and $\sigma$ can be computed iteratively as follows. Using the notation from Algorithm~\ref{alg:greedy-envelope}, we only store the submatrix $M$, which is initially set to $W$. In each step we find the sparsest column of $M$, which will be the next column in the column permutation, and the rows containing the non-zero elements of this column, which will be the next rows in the row permutation. Then we simply delete the non-zero elements in the chosen column and rows and iterate the procedure. In order to find the sparsest column in each iteration, we maintain a minheap storing for each column the number of non-zero elements. The smallest element can be removed in time $\mathcal{O}(n)$ which yields the sparsest column. Then we can delete the elements as described above, each in time $\mathcal{O}(n)$. For every deleted element we decrease the number of non-zero elements in the containing column by one, and thus we may have to reorder the minheap. But this can also be done in time $\mathcal{O}(n)$. The procedure stops when all non-zero elements have been deleted after total time $\mathcal{O}(n\, \nnz(W))$.

\section{Results for sparse state preparation in practice}
\label{sec:numerical-results}

We compare the \cnot{} counts resulting from our sparse state preparation scheme presented in Section~\ref{sec:ssp} to the dense case. The implementation of the sparse state preparation scheme is described in Section~\ref{sec:ssp_impl}. In order to improve the classical computation time, we do not consider all possible qubit splittings, but randomly sample $100$ splittings and choose the one with the largest number of non-zero elements in one column. We use the dense state preparation scheme from~\cite{plesch}, which achieves near optimal \cnot{} counts for arbitrary dense states. All of the methods are implemented in {\sc UniversalQCompiler}~\cite{compiler}. The results are presented in Figure~\ref{fig:ssp}. These indicate the advantage we gain by taking into account the sparseness. Note however, that the dense case outperforms the sparse case for fairly dense states (where the cost of pivoting is not compensated by the smaller state preparation). It is also worth noting that the counts found in practice are significantly smaller than our upper bounds.

\begin{figure}[h]
\includegraphics[width=0.5\textwidth]{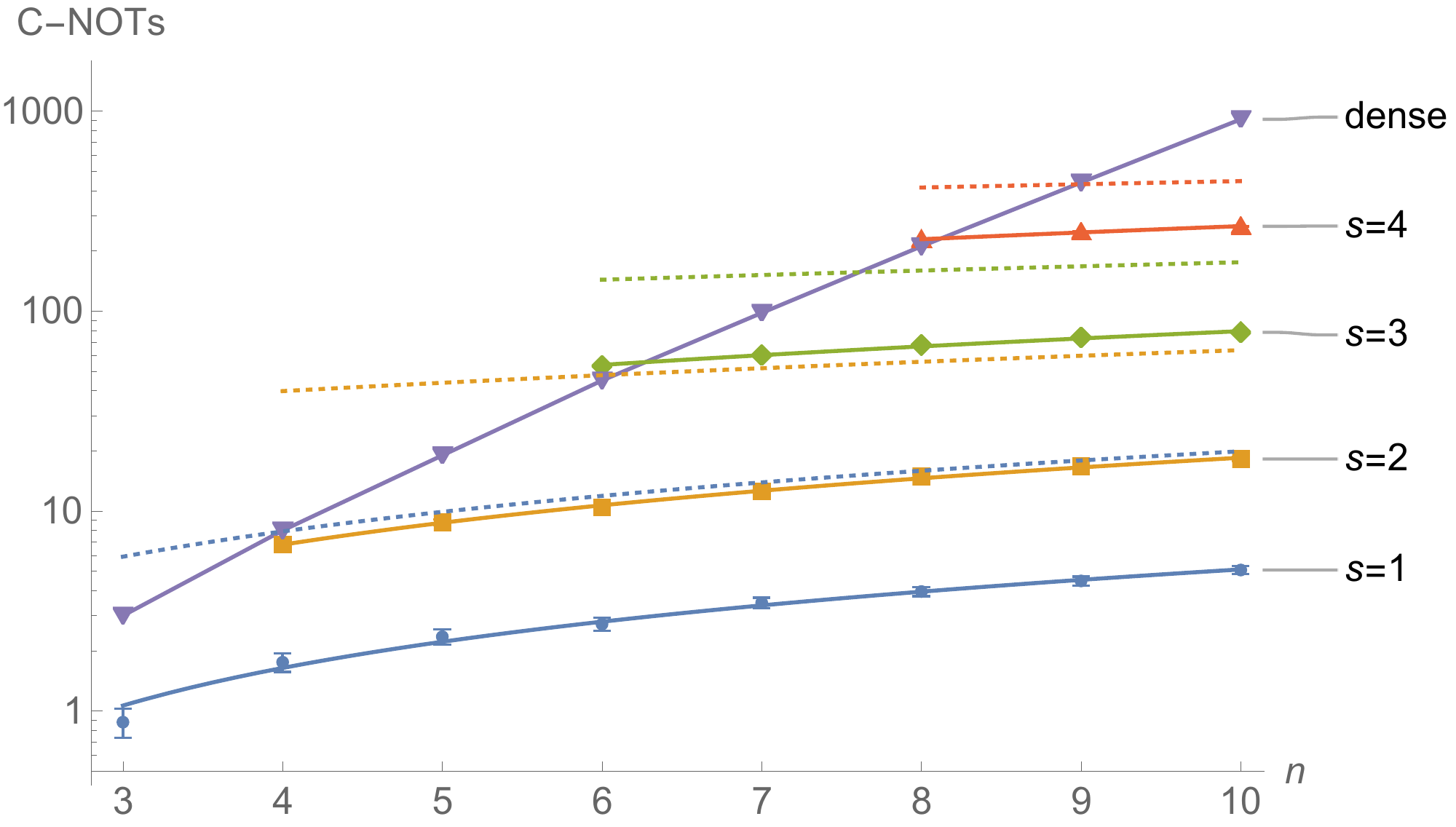}
\caption{{\bf Performance of sparse state preparation.}  This shows the average number of \cnot{} gates produced by our sparse state preparation scheme for sparse states on $n$ qubits with $2^s$ non-zero entries, whose positions are chosen uniformly at random.  Note that the actual values of the non-zero entries do not influence the counts. We used $200$ trials and the error bars indicate twice the standard error of the mean (about 95\% confidence). Note that for $s\geq2$ the error bars are too small to be visible. We added linear best fit lines (except in the dense case). No ancillas were used for our implementation. The dashed lines show the upper bound of $(n+6s-7+23/24)2^s$ from Table~\ref{tab:explicit-counts}, based on $\lceil\tfrac{s}{2}-1\rceil$ clean ancillas. Although our implementation does not use ancillas, it nevertheless beats this bound.}
\label{fig:ssp}
\end{figure}

\acknowledgements
The authors thank Vadym Kliuchnikov for pointing out Ref.~\cite{vadym}.

RC is supported by EPSRC's Quantum Communications Hub (grant numbers EP/M013472/1 and EP/T001011/1).
RI acknowledges support from the Swiss National Science Foundation through SNSF project No. 200020-165843 and through the National Centre of Competence in Research \textit{Quantum Science and Technology} (QSIT). 

\providecommand{\noopsort}[1]{}\providecommand{\singleletter}[1]{#1}%

\appendix

\section{Details on Householder reflections}
\label{A:vector}

Given two unit vectors $\ket{v}$ and $\ket{w}$ we want to construct a map that sends $\ket{v}$ to $\e^{\ii\theta}\ket{w}$ for some real $\theta$. This can be achieved using a Householder reflection with respect to the unit vector
$$\ket{u_\theta}=\frac{\ket{v}-\e^{\ii\theta}\ket{w}}{\|\ket{v}-\e^{\ii\theta}\ket{w}\|}.$$
Let us first compute the normalization
\begin{align*}
\|\ket{v}-\e^{\ii\theta}\ket{w}\|^2
&= (\bra{v}- \e^{-\ii\theta}\bra{w})(\ket{v}- \e^{\ii\theta}\ket{w}) \\
&= 2 - 2 \Re (\e^{\ii\theta}\ip{v}{w}) \\
&= 2 \Re (1-\e^{\ii\theta}\ip{v}{w}) \\
&= 2 \Re (z).
\end{align*}
where $z:=1-\e^{\ii\theta}\ip{v}{w}$.  If $z=0$ no transformation is needed, so assume $z\neq0$.

Consider the generalized Householder reflection $H_{u_\theta}^\phi$ acting on $\ket{v}$. We obtain
\begin{align*}
H_{u_\theta}^\phi \ket{v}
&= \ket{v} + \frac{\e^{\ii\phi} - 1}{2\Re (z)} (\ket{v}- \e^{\ii\theta}\ket{w})(1-\e^{-\ii\theta}\ip{w}{v})\\
&= \ket{v} + \frac{\e^{\ii\phi} - 1}{2} \frac{\bar{z}}{\Re (z)} (\ket{v}- \e^{\ii\theta}\ket{w}).
\end{align*}
Requiring $H_{u_\theta}^\phi \ket{v} \propto \ket{w}$  leads to the following condition
$$\frac{\e^{\ii\phi} - 1}{2} \frac{\bar{z}}{\Re (z)} = -1 \Leftrightarrow \e^{\ii\phi} = 1-\frac{\bar{z} + z}{\bar{z}} = \frac{-z}{\bar{z}}$$
which we can also write as $\phi = \pi+2\arg(z) \mod 2\pi$.

Note that numerical instabilities arise when the norm of $\ket{v}-\e^{\ii\theta}\ket{w}$ is very close to zero. In such cases we can either increase the precision of the computation, or note that for very small norms ignoring the rotation is equivalent to allowing a small error.

\subsection{Standard Householder reflection}

For a standard Householder reflection we have $\phi = \pi$. Now choose $\theta=\pi-\arg(\ip{v}{w})$ or $\theta=0$ if $\ip{v}{w}=0$. This implies $z=1+|\ip{v}{w}|\neq 0$. Then we define
$$H_{v,w}=H_{u_\theta}.$$
This map has the property that $H_{v,w}\ket{v}=\e^{\ii\theta}\ket{w}$. 

Note that in this case $z$ is never close to $0$. Since the decompositions presented in the main text only use standard Householder reflections, we do not have to deal with numerical instabilites there.

\subsection{Generalized Householder reflection}

If we want to get rid of the phase $\e^{\ii\theta}$ we have to use a generalized Householder reflection. Setting $\theta=0$ implies that $z=1-\ip{v}{w}$ which might lead to numerical instabilities. Then $\phi=\pi+2\arg(z)$. We define
$$\tilde{H}_{v,w}=H_{u_\theta}^\phi,$$
so that $\tilde{H}_{v,w}\ket{v}=\ket{w}$.

\section{Proofs for the dense Householder decomposition}
\label{app:dense-householder}

Here we give the proofs of Lemma~\ref{lemma:dense-isometries} and Corollary~\ref{coro:dense-unitaries} stated in the main text. 

\begin{proof}[Proof of Lemma~\ref{lemma:dense-isometries}]
  This follows analogously to the proof of Theorem~1 of~\cite{isometries}, so we only point out the necessary modifications. The idea is to decompose the isometry via
  \begin{equation}\label{eq:dense_idea}
    V=S_{2^m-1}H_0^{\phi_{2^m-1}}S_{2^m-1}^\dagger\ldots S_0H_0^{\phi_0}S_0^\dagger\,,
  \end{equation}
  where $S_i$ implements state preparation of the $i^{\text{th}}$ column of $V$. Our modification is to replace the $2^m$ generalized Householder reflections by standard Householder reflections, and then to correct for the difference by using a diagonal gate on $m$ qubits at the end, i.e.,
  \begin{equation}\label{eq:dense}
    V=(\Delta_m\ot I_{n-m}) S_{2^m-1}H_0S_{2^m-1}^\dagger\ldots S_0H_0S_0^\dagger\,.
  \end{equation}
Lemma~\ref{lemma:reflections} shows that using one dirty ancilla qubit the standard Householder reflection $H_0$ can be implemented using $16n-24$ \cnots{} instead of $16n^2-60n+42$ used for the generalized version $H_0^{\phi}$. The cost of the diagonal gate is $2^m-2$ (see Table~\ref{tab:counts}).
\end{proof}

\begin{proof}[Proof of Corollary~\ref{coro:dense-unitaries}]
  First we claim that a controlled isometry from $m-1$ to $m$ qubits with $n-m$ controls can be implemented using at most $N(m-1,m)+16(n-m)2^{m-1}+2^{n-1}-2^{m-1}$ \cnots{} where $N(m,n)$ denotes the upper bound for implementing an isometry from $m$ to $n$ qubits given in Lemma~\ref{lemma:dense-isometries}\footnote{We write $N$ rather that $\mathcal{N}_{\mathrm{iso}}$ to emphasise that the modification is only valid based on the technique of Lemma~\ref{lemma:dense-isometries}.}. This follows from Lemma~\ref{lemma:dense-isometries} and the fact that when implementing the controlled isometry by controlling all the gates in~\eqref{eq:dense}, we do not need to control the state preparation gates or their inverses (i.e., the $S_i$ and $S_i^\dagger$ gates do not need controls). Thus the control does not affect the first three terms in the counts from Lemma~\ref{lemma:dense-isometries}.

  To decompose the unitary, we start by reducing the first half of the columns using the inverse of a circuit implementing an isometry from $n-1$ to $n$ qubits. This yields $\ket{0}\!\!\bra{0}\ot I+\ket{1}\!\!\bra{1}\ot U_1$, where $U_1$ is an $n-1$ qubit unitary. We can then control on the first qubit and do the inverse of an isometry from $n-2$ to $n-1$ qubits and so on. At each step we reduce half of the remaining columns and requires the inverse of an isometry from $n-k-1$ to $n-k$ qubits with $k$ controls, where $k=0,1,\ldots,n-1$.

  The \cnot{} count is thus
$$\sum_{k=0}^{n-1}N(n-k-1,n-k)+k2^{n-k+3}+2^{n-1}(1-2^{-k})\,.$$
This can be used with the counts from Lemma~\ref{lemma:dense-isometries} to upper bound the number of \cnots{}.
To generate a clean bound on the leading order term, note that if $n$ is even, the even values of $k$ contribute $\tfrac{23}{48}\sum_{t=0}^{(n-2)/2}2^{-4t}$ to the coefficient of $4^n$, while the odd values of $k$ contribute $\tfrac{115}{768}\sum_{t=0}^{(n-2)/2}2^{-4t}$, giving a leading order of $\tfrac{161}{240}4^n$.
Similarly, if $n$ is odd, the even values of $k$ contribute $\tfrac{115}{192}\sum_{t=0}^{(n-2)/2}2^{-4t}$ to the coefficient of $4^n$, while the odd values of $k$ contribute $\tfrac{23}{192}\sum_{t=0}^{(n-2)/2}2^{-4t}$, giving a leading order of $\tfrac{23}{30}4^n$.
\end{proof}

\section{Multi-controlled \nt{} gates}

We denote a $k$-controlled \nt{} gate by $C_{k}(X)$. Using a single dirty ancilla qubit such gates can be decomposed with a linear number of \textsc{C-not} gates. We start by recalling two lemmas from~\cite{barenco, isometries}.
\begin{lemma}(\cite[Lemma 7.3]{barenco}) \label{lemma:toff2}
Let $n \geq 5$ denote the total number of qubits. A $C_{n-2}(X)$ gate can be decomposed into two $C_{k}(X)$ and two $C_{n-k-1}(X)$, where $k\in\{2,3,\dots,n-3\}$.
\end{lemma}
\begin{lemma}(\cite[Lemma 8]{isometries}) \label{lemma:toff1}
Let $n \geq 5$ denote the total number of qubits and $k \in \{1,\dots, \ceil{\frac{n}{2}}\}$, then we can implement a $C_{k}(X)$ gate with at most 8k-6 \textsc{C-not}s.
\end{lemma}
Note that if $k\geq5$ this bound can be reduced to $8k-12$~\cite{toffoli-ancilla}. However, we do not use this here for the convenience of having a single bound for all $k\leq\lceil\tfrac{n}{2}\rceil$. The desired decomposition of multi-controlled \nt{} gates using a single ancilla qubit follows.
\begin{corollary} \label{cor:toffoli}
Let $n \geq 3$ denote the total number of qubits. Then we can implement a $C_{n-2}(X)$ gate with at most $16n-40$ \textsc{C-not}s.
\end{corollary}
\begin{proof}
The case $n=3$ is trivial and $n=4$ can be done with 6 \cnots~\cite[Sec.~VI A]{barenco}. For $n\geq5$ we can use Lemma~\ref{lemma:toff2}, choosing $k=\lfloor\frac{n}{2}\rfloor$ (note that $2\leq k\leq n-3$). We then have both $k\leq\lceil\frac{n}{2}\rceil$ and $n-k-1\leq\lceil\frac{n}{2}\rceil$, so can use Lemma~\ref{lemma:toff1}.  This gives a count of at most
$$2(8k-6)+2(8(n-k-1)-6)=16n-40\,,$$
as required.
\end{proof}

\section{Permutation gates}\label{app:perm}
One key feature of several of our methods is the use of permutations to adjust the form of the given isometry. Here we discuss the number of \cnots{} needed for these.

\begin{lemma}
A permutation gate on $n\geq3$ qubits can be performed without ancilla using $(27n-62)2^n+44n^2-96n-23$ \cnots.
\end{lemma}
\begin{proof}
  Permuting the rows of a state on $n$ qubits corresponds to constructing a $2^n\times2^n$ permutation matrix (i.e., a unitary matrix with one 1 in every row and column). Each such matrix corresponds to a permutation on $2^n$ objects. It is known that all permutations can be decomposed as a sequence of swaps. A permutation is \emph{even} if it can be decomposed into an even number of swaps and otherwise it is \emph{odd}. It is known that all even permutations on $n\geq3$ qubits can be performed with at most $n$ \nts{}, $n^2$ \cnots{} and $3(2^n+n+1)(3n-7)$ Toffoli gates~\cite[Theorem~33]{reversible}. A Toffoli gate can be performed up to a diagonal gate using $3$ \cnots{}~\cite{barenco} and diagonal gates can be commuted with \nt{}, \cnot{} and Toffoli gates up to another diagonal.  Thus, ignoring the single-qubit (\nt{}) gates, any even permutation can be decomposed into at most $(27n-63)2^n+28n^2-36n-63$ \cnots{} without ancilla, up to a diagonal gate.

  If we have an odd permutation on $n$ qubits, we can apply an $(n-1)$-controlled \nt{} to make it an even permutation.  Without an ancilla, we can do this as a $n-1$ controlled single-qubit unitary with an overhead of $16n^2-60n+42$ \cnots{} (see Table~\ref{tab:counts}). This leads to an overall \cnot{} count of $(27n-63)2^n+44n^2-96n-21$ \cnots{} without ancilla\footnote{Slightly lower counts are possible with ancilla but these do not change the leading order so we do not consider them here for simplicity.}.

  A diagonal gate on $n$ qubits can be performed using $2^n-2$
  \cnots{} (see Table~\ref{tab:counts}), leading to an overall \cnot{}
  count of $(27n-62)2^n+44n^2-96n-23$.
\end{proof}
Slightly lower counts are possible with ancillas, but these do not change the leading order so we do not consider them here for simplicity.  Note also that the above bound is always less than $27n2^n$ for $n\geq3$, so we use the latter as a simplification.

Any unitary on $n\geq2$ qubits can be decomposed using at most $\lceil\frac{23}{48}4^n-\frac{3}{2}2^n+\frac{4}{3}\rceil$ \cnots~\cite{diagonal}
(without ancilla) which gives a better count for an arbitrary permutation for $n\leq8$.

[In~\cite{reversible}, the authors note that if we add a qubit on which we do not act, an odd permutation becomes even and hence any permutation on $n\geq3$ qubits can be done using one ancilla and an even permutation on $n+1$ qubits~\cite[Corollary~13]{reversible}. However, this is significantly worse than the above bound.]

\begin{lemma}
A permutation gate on $n\geq2$ qubits can be performed with one dirty ancilla and at most $(18n-26)(2^n-1)$ \cnots.
\end{lemma}
\begin{proof}
  The idea is to use Householder reflections.  We can write the permutation gate as $\sum_{i=0}^{2^n-1}\ket{j(i)}\!\!\bra{i}$, where $\ket{j(i)}$ is a computational basis state. The Householder reflection $H_{j(i),i}$ takes $\ket{j(i)}\mapsto\ket{i}$ and $\ket{i}\mapsto\ket{j(i)}$ without affecting any other columns (cf.\ Corollary~\ref{corr:fill-in}).

  Since $H_{j(i),i}=H_u$, where $\ket{u}=\frac{1}{\sqrt{2}}\left(\ket{i}-\ket{j(i)}\right)$, we can implement each Householder reflection along the lines given in Lemma~\ref{lemma:sp}.  The state $\ket{u}$ can be reduced to $\ket{i}$ as follows.  Let $\ket{i}=\ket{i_1i_2\ldots i_n}$ and $\ket{j(i)}=\ket{j_1j_2\ldots j_n}$ in the computational basis. Find an index $k$ such that $j_k\neq i_k$. Controlling on the $k^{\text{th}}$ qubit we can apply at most $n-1$ \cnots{} such that $\ket{j(i)}\mapsto\ket{i_1\ldots i_{k-1}j_ki_{k+1}\ldots i_n}$ and $\ket{i}$ is unchanged.  When applied to $\ket{u}$, this results in $\ket{u'}=\frac{1}{\sqrt{2}}\left(\ket{i}-X_k\ket{i}\right)$, where $X_k$ is a \nt{} on the $k^{\text{th}}$ qubit.  Applying a single qubit rotation to the $k^{\text{th}}$ qubit then maps $\ket{u'}$ to $\ket{i}$. Let us denote the reverse of these steps by $\SP_{i,u}$, so that $\SP_{i,u}\ket{i}=\ket{u}$.  We have
  $$H_{\ket{u}}=\SP_{i,u}^{\phantom{\dagger}}H_i\SP_{i,u}^\dagger\,,$$
  where $H_i=I-2\ket{i}\!\!\bra{i}$ is either a $Z$ or $-Z$ gate with $n-1$ controls, and hence has the \cnot{} count of $C_{n-1}(X)$, which is $16n-24$ if one dirty ancilla is available (see Table~\ref{tab:counts}).

  It follows that we can do $H_{j(i),i}$ using $18n-26$ \cnots{}, and hence the whole permutation matrix using at most $(18n-26)(2^n-1)$ \cnots{}.
\end{proof}
Again, for simplicity, we can bound this by $(18n-26)2^n$.

\section{Knill-Householder decomposition} \label{app:knill-householder}
In this appendix we consider a generalized decomposition scheme that contains Knill's decomposition~\cite{knill} and the Householder decomposition as special cases. Suppose $U$ is a unitary that we want to implement and $B$ is another unitary representing a change of basis. Assume that $B$ and $\tilde{U} = B^\dagger UB$ are sparse. Let $\ket{b_i} = B\ket{i}$, then $\ket{\tilde{u}_i} := UB\ket{i} = U\ket{b_i}$ and
$$U=\sum_i U\op{b_i}{b_i} = \sum_i \op{\tilde{u}_i}{b_i}.$$
Consider the generalized Householder reflection $\tilde{H}_{\tilde{u}_0,b_0}$ mapping $\ket{\tilde{u}_0}$ to $\ket{b_0}$ (in this section we want this map to be exact and not up to a phase). In the basis formed by the columns of $B$ this reduces the first column of $\tilde{U}$. Consider the effect on the second column $\ket{\tilde{u}_1}$. Since it is orthogonal to $\ket{\tilde{u}_0}$, this column will suffer fill-in if and only if it is not orthogonal to $\ket{b_0}$. In this case, fill-in will be confined to the subspace spanned by $\ket{\tilde{u}_0}$ and $\ket{b_0}$. This means that fill-in is determined by the structure of $\tilde{U}$ and works in the same way as for the Householder reflection in the standard basis. It is important to note however that state preparation is only efficient if $\ket{\tilde{u}_0}$ is sparse, i.e., if both $\tilde{U}$ and $B$ are sufficiently sparse. Continuing in this fashion allows us to reduce $U$ to the identity.

\begin{remark}
If we define $B$ to be the unitary whose columns form an eigenbasis of $U$, then this scheme reduces to Knill's decomposition, and if $B$ is the identity, it is essentially the Householder decomposition. Note however that in the case of the Householder decomposition we used standard Householder reflections and it was sufficient to implement them up to diagonal and permutation gates.
\end{remark}

This method can be generalized to isometries by extending a given isometry $V$ to a unitary $U$. This can be done such that $U$ has at least $2^n-2^m$ eigenvalues equal to 1 (see~\cite{knill}, \cite[Lemma~5]{compiler}). Let $B$ be an $n$ qubit unitary whose first $2^n-2^m$ columns are eigenvectors of $U$ with eigenvalue 1. Then $\tilde{U}=B^\dagger UB$ is blockdiagonal with a $(2^n-2^m)\times(2^n-2^m)$ trivial block and an $2^m\times2^m$ non-trivial block. This observation can be used to reduce $U$ as described above.

\end{document}